\documentclass[aps,pra,twocolumn,notitlepage,superscriptaddress,longbibliography,floatfix,10pt]{revtex4-2}

\usepackage{amsmath,amssymb,amsthm,mathtools}
\usepackage{bm}
\usepackage{graphicx}
\usepackage{booktabs}
\usepackage{xcolor}
\usepackage{enumitem}
\usepackage{microtype}

\theoremstyle{definition}

\theoremstyle{plain}
\newtheorem{theorem}{Theorem}

\newtheorem{proposition}{Proposition}
\newtheorem{corollary}{Corollary}

\newcommand{\loc}{\mathcal L}
\newcommand{\ns}{\mathcal N\!S}
\newcommand{\sat}{\mathcal S}
\newcommand{\qtwo}{\mathcal Q_{2}}
\newcommand{\qfour}{\mathcal Q_{4}}
\newcommand{\bits}{\,\mathrm{bits}}

\usepackage{natbib}
\usepackage[hidelinks]{hyperref}
\begin{document}

\title{Complexity-Aware Theory Testing from Bell Witnesses}

\author{Jianshuo GAO}
\affiliation{Peking University, School of Physics}
\date{\today}

\begin{abstract}
Bell statistical-strength analyses and complexity-based model selection are usually treated separately.
Here we relate them by showing that a witness obtained from a coarse-graining of full Bell trials yields, through data processing, a lower bound on the Kullback--Leibler (KL) distance to a competitor class in terms of the induced witness distribution.
For binary Bell-game witnesses this reduces to a Bernoulli bound, and in the CHSH scenario the local image collapses to a single threshold, giving the closed-form expression
$D_{\mathrm{KL}}(\mathrm{Bern}(\omega)\Vert \mathrm{Bern}(3/4))$ under uniform inputs, with a corresponding extension to known nonuniform designs.
A finite-sample Hoeffding argument gives a lower confidence bound under independent trials.
We also include a non-CHSH example based on the three-party Mermin--GHZ game.
Because the bound is measured in bits per trial, it can be compared directly with an MDL/BIC-type complexity penalty and thereby yields a conservative crossover criterion for when a more expressive competitor becomes worthwhile.
For the reproducible four-photon data of Wang \emph{et al.}, the witness certifies a positive information gap against locality, while a full-table comparison across local, no-signalling, saturated, and two compact nonlocal families favors low-dimensional nonlocal descriptions once complexity is charged.
A four-parameter unbiased-correlator control shows that the data support compact nonlocality over locality, while only weakly distinguishing the specific cosine structure of the two-parameter model; an AIC comparison instead favors broader nonlocal controls.
We also report witness-based benchmarks from additional published CHSH experiments and discuss the interpretational scope of BIC for constrained or non-regular model classes.
\end{abstract}

\maketitle

\section{Introduction}

How should finite experimental data shift our preference from a simpler physical theory to a more expressive one?
In quantum foundations this question is often posed in Bell form: how strongly do the data disfavor locality?
In statistical model selection the same question appears as a trade-off between fit and complexity, formalized by minimum description length (MDL), BIC-type approximations to Bayesian evidence, or related Occam penalties \cite{rissanen1978modeling,barron1998minimum,grunwald2007minimum,schwarz1978estimating,drton2017singular}.
These viewpoints are closely related in spirit, but they are usually operationalized at different levels.
Bell analyses often report a witness value, a $p$ value, or a KL rate against local realism; model-selection analyses compare fully fitted classes through a code-length or marginal-likelihood proxy.
What is still missing is a direct bridge from an experimentally visible witness to a complexity-charged model-comparison criterion.

The aim of this paper is to make that connection explicit.
The individual ingredients are standard---KL divergence in Bell statistical-strength analyses and MDL/BIC-style penalties in model selection---but their combination leads to a useful inference rule.
A directly observed Bell witness can be converted into a certified information rate, measured in bits per trial, and this quantity can be set against a complexity penalty written in the same units.
This yields a conservative crossover criterion for when a broader class is justified, even before a full-table fit of that class is carried out.

The CHSH scenario provides the most transparent closed-form realization of this idea.
If a trial is coarse-grained to the usual win/loss variable, then data processing implies that the KL distance from the full trial distribution to the local polytope is at least the Bernoulli KL distance between the observed win rate and the best local win rate.
Under uniform inputs this reduces to the explicit lower bound
\begin{equation}
D_{\mathrm{KL}}(P\Vert \loc)
\ge
D_{\mathrm{KL}}\!\left(\mathrm{Bern}(\omega(P))\,\middle\|\,\mathrm{Bern}(3/4)\right).
\end{equation}
Because the quantity on the right is measured in bits per trial, it can be compared directly with an MDL/BIC penalty.

The same mechanism is not restricted to CHSH.
Any binary Bell-game witness gives a Bernoulli KL certificate, including multipartite and multi-outcome games that admit a win/loss score.
More generally, any finite witness obtained by coarse-graining the full trial alphabet inherits a witness-level KL lower bound.
CHSH is convenient because the induced local witness family collapses to a one-parameter Bernoulli half-interval with a closed-form boundary.

The second issue is finite data.
We therefore derive a finite-sample confidence version using Hoeffding's inequality \cite{hoeffding1963probability}.
The resulting lower confidence bound is rigorous for independent trials with a known measurement design, whether the settings are randomized shot by shot or fixed by an external schedule.
When only reconstructed coincidence tables are available, as in the case study below, the finite-sample interpretation should be regarded as an effective-trial approximation rather than as a literal application of the theorem.

The third issue concerns model-selection interpretation.
The local and no-signalling polytopes are constrained boundary models, so the regularity assumptions behind classical Schwarz-type BIC do not apply without qualification \cite{watanabe2009algebraic,drton2017singular}.
Accordingly, we use BIC as a transparent MDL proxy rather than as an exact asymptotic evidence formula for non-regular classes.
To keep that distinction visible, we supplement the BIC comparison with a low-dimensional control family, an AIC comparison, and simulation-based calibration.

The paper proceeds in three steps.
First, we prove a general witness coarse-graining theorem and a binary Bell-game corollary, and then specialize them to CHSH with an explicit design-dependent local benchmark.
Second, we relate the resulting witness-certified KL rate to a complexity penalty measured in bits per trial, which leads to a conservative crossover criterion.
Third, we examine how this relation behaves in practice through a non-CHSH example based on the three-party Mermin--GHZ game, a detailed reanalysis of the Wang \emph{et al.} data \cite{wang2025violation}, witness-based checks on additional published CHSH experiments, and simulation and AIC-based comparisons that clarify which empirical conclusions are stable across complexity proxies.

\section{Relation to prior work}

Three literature threads meet here.
The first is the Bell statistical-strength program, where KL divergence quantifies the asymptotic evidence rate against local realism \cite{vandam2005statistical} and can be optimized over source states, settings, or detector efficiencies \cite{zhang2010statistical,zhang2013efficient}.
That literature already teaches us that KL divergence is the right evidence metric.
What it does \emph{not} directly provide is a witness-to-complexity bridge: given only a coarse witness value, when does that witness certify enough evidence to justify a broader model class once complexity is charged?

The second thread is MDL, Bayesian evidence, and algorithmic Occam reasoning \cite{rissanen1978modeling,barron1998minimum,li2008introduction,grunwald2007minimum}.
These frameworks formalize the idea that likelihood gains must be weighed against the code length or dimension of the model that achieves them.
But they typically start from a fully specified model class and a fitted likelihood.
They do not usually begin with a Bell witness and ask how much complexity budget that witness alone can already certify.

The third thread concerns the limitations of classical BIC in singular or constrained models.
Polytope models can live on boundaries, corners, or lower-dimensional active faces, so the regular asymptotics underlying Schwarz's derivation need not apply directly \cite{schwarz1978estimating,watanabe2009algebraic,drton2017singular}.
This matters acutely in Bell scenarios, where local and no-signalling sets are polyhedral.
Our approach is therefore to keep the witness theorems separate from the choice of full-table complexity proxy, and to discuss the scope of the latter explicitly in Sec.~\ref{sec:bicscope}.
From this point of view, the Bell side contributes an evidence metric and the MDL side contributes a complexity accounting; the link between them is a witness-based KL rate that can be compared directly with a penalty written in bits per trial.
This comparison can be carried out before a large alternative class has been fully fitted.

\section{Operational framework: complexity certificates and crossover rules}

We fix an experimental description language $\mathcal E$ and represent a theory $T$ by an executable predictor that returns the probabilities relevant to each experiment $e\in\mathcal E$.
An explicit reference implementation $d_T$ serves as a \emph{complexity certificate}; its length
\begin{equation}
L_{\mathrm{cert}}(T):=|d_T|
\end{equation}
upper-bounds the ideal predictive description length up to an additive constant \cite{li2008introduction}.
For fitted parametric classes we will later use BIC as a transparent MDL proxy.

This coding viewpoint induces the standard Occam factor.
If candidate theories $T_i$ receive prior weights $w_i\propto 2^{-L_i}$, then posterior odds factor into a complexity term and a likelihood term.

\begin{proposition}[Posterior odds = Occam factor $\times$ Bayes factor]
\label{prop:posteriorodds}
Let $\{T_i\}$ be candidate predictors with likelihoods $\nu_i(\cdot)$ and code lengths $L_i$.
For any dataset $x$,
\begin{equation}
\frac{\mathbb P(T_i\mid x)}{\mathbb P(T_j\mid x)}
=
2^{-(L_i-L_j)}\,
\frac{\nu_i(x)}{\nu_j(x)}.
\end{equation}
\end{proposition}

The remaining question is how to estimate the likelihood term from witness data.
The next theorem isolates the asymptotic logic.

\begin{theorem}[KL advantage eventually overcomes a fixed complexity penalty]
\label{thm:mdl}
Let $p$ and $q$ be computable distributions on a finite alphabet, and let $Y_1,Y_2,\dots\sim p$ i.i.d.
Define the two-part MDL score
\begin{equation}
\mathrm{MDL}_n(r):=K(r)-\log_2 r^{\otimes n}(Y_{1:n}).
\end{equation}
If $D_{\mathrm{KL}}(p\Vert q)=\delta>0$, then with probability $1$ there exists $N$ such that for all $n\ge N$,
\begin{equation}
\mathrm{MDL}_n(p)<\mathrm{MDL}_n(q).
\end{equation}
Moreover, the crossover scale is governed by
\begin{equation}
n\,\delta \gtrsim K(q)-K(p).
\label{eq:abstractcrossover}
\end{equation}
\end{theorem}

\begin{proof}
The log-likelihood ratio
$\sum_{i=1}^n \log_2[p(Y_i)/q(Y_i)]$
has mean $n\delta$.
By the strong law of large numbers this grows linearly in $n$, whereas the coding penalty is $O(1)$.
Hence the MDL score of $p$ eventually beats that of $q$.
\end{proof}

For practical parametric fits we will use the familiar proxy
\begin{equation}
L_{\mathrm{BIC}}(\mathcal M;x_{1:n})
:=
-\log_2 \hat{\mathcal L}_{\mathcal M}(x_{1:n})
+\frac{d_{\mathcal M}}{2}\log_2 n,
\label{eq:bicdef}
\end{equation}
where $d_{\mathcal M}$ is an effective dimension.
If a witness supplies a conservative evidence rate $\delta_{\mathrm{wit}}$ in bits per trial against a simpler reference class $\mathcal M_{\mathrm{ref}}$, then a more expressive competitor $\mathcal M$ is conservatively justified once
\begin{equation}
n\,\delta_{\mathrm{wit}}
\gtrsim
\frac{d_{\mathcal M}-d_{\mathcal M_{\mathrm{ref}}}}{2}\log_2 n.
\label{eq:practicalcrossover}
\end{equation}
If the competitor is actually \emph{simpler} than the reference class, as happens below for the two-parameter cosine family versus the eight-dimensional local polytope, then Eq.~\eqref{eq:practicalcrossover} does not need to pay an extra complexity tax at all.
The crossover logic is most informative when the alternative class is broader, such as the saturated control versus locality.

\section{Witness-certified KL gaps}

\subsection{General witness coarse-graining theorem}

Let $Z$ denote the full trial alphabet and let $\mathcal C$ be any competitor class of distributions on $Z$.
A witness is a measurable coarse-graining $\pi:Z\to W$ to a finite alphabet $W$.
Write $\pi_\# P$ for the induced witness distribution and $\pi_\#\mathcal C:=\{\pi_\# Q:Q\in\mathcal C\}$ for the image of the competitor class.

\begin{theorem}[Witness coarse-graining lower bound]
\label{thm:coarsegraining}
For every distribution $P$ on $Z$,
\begin{equation}
\inf_{Q\in\mathcal C} D_{\mathrm{KL}}(P\Vert Q)
\ge
\inf_{R\in\pi_\#\mathcal C} D_{\mathrm{KL}}(\pi_\# P\Vert R).
\label{eq:generalcoarsegrain}
\end{equation}
\end{theorem}

\begin{proof}
For any fixed $Q\in\mathcal C$, data processing for KL divergence gives
$D_{\mathrm{KL}}(P\Vert Q)\ge D_{\mathrm{KL}}(\pi_\# P\Vert \pi_\# Q)$.
Taking the infimum over $Q\in\mathcal C$ yields Eq.~\eqref{eq:generalcoarsegrain}.
\end{proof}

The theorem is general.
Its usefulness depends on how explicitly one can characterize the image set $\pi_\#\mathcal C$.
CHSH is convenient because the win/loss coarse-graining makes the image of the local class one-dimensional.
But the theorem itself is not limited to CHSH, to two outcomes, or even to two parties.

A particularly important special case is any binary Bell game.

\begin{corollary}[Binary Bell-game witnesses]
\label{cor:binarygame}
Let $W\in\{0,1\}$ be the win indicator of any Bell game, and let
\begin{equation}
\omega_{\mathrm{loc}}:=\sup_{Q\in\loc}\mathbb E_Q[W]
\end{equation}
be the local value of the game.
Then for every observed distribution $P$,
\begin{equation}
\inf_{Q\in\loc} D_{\mathrm{KL}}(P\Vert Q)
\ge
D_{\mathrm{KL}}\!\left(\mathrm{Bern}(\omega(P))\,\middle\|\,\mathrm{Bern}(\omega_{\mathrm{loc}})\right),
\label{eq:binarygamebound}
\end{equation}
where $\omega(P):=\mathbb P_P(W=1)$.
\end{corollary}

\begin{proof}
Apply Theorem~\ref{thm:coarsegraining} with $\pi$ equal to the binary win/loss map.
The image of any local model is a Bernoulli law with success parameter at most $\omega_{\mathrm{loc}}$.
\end{proof}

Corollary~\ref{cor:binarygame} already covers many witnesses beyond CHSH, including multipartite or multi-outcome nonlocal games with binary scoring.
For more general linear Bell functionals one can still apply Theorem~\ref{thm:coarsegraining}, but the push-forward image $\pi_\#\loc$ becomes higher-dimensional and the closed-form simplification is lost.
Thus CHSH is special mainly because the coarse witness image is so tractable.

\subsection{Worked example beyond CHSH: the three-party Mermin--GHZ game}

To make the generality claim concrete, consider the standard three-party GHZ game \cite{mermin1990extreme}.
Inputs $(x,y,z)\in\{0,1\}^3$ are drawn uniformly from the even-parity set
\begin{equation}
\mathcal X_{\mathrm{GHZ}}:=\{000,011,101,110\},
\end{equation}
and outputs $(a,b,c)\in\{0,1\}^3$ win when
\begin{equation}
a\oplus b\oplus c = x\lor y\lor z.
\end{equation}
Equivalently, the parity of the outputs should be $0$ for input $000$ and $1$ for the other three allowed inputs.
A deterministic local strategy cannot satisfy all four parity constraints simultaneously, because multiplying the first three constraints forces the fourth to fail.
Hence every local strategy wins on at most three of the four allowed inputs, so the local value is $\omega_{\mathrm{loc}}^{\mathrm{GHZ}}=3/4$.

\begin{proposition}[Mermin--GHZ Bernoulli certificate]
\label{prop:ghzgame}
Let $W_{\mathrm{GHZ}}\in\{0,1\}$ be the win indicator of the three-party GHZ game under the uniform input distribution on $\mathcal X_{\mathrm{GHZ}}$.
For every observed distribution $P$ on inputs and outputs,
\begin{equation}
\inf_{Q\in\loc} D_{\mathrm{KL}}(P\Vert Q)
\ge
D_{\mathrm{KL}}\!\left(\mathrm{Bern}(\omega_{\mathrm{GHZ}}(P))\,\middle\|\,\mathrm{Bern}(3/4)\right),
\end{equation}
where $\omega_{\mathrm{GHZ}}(P):=\mathbb P_P(W_{\mathrm{GHZ}}=1)$.
\end{proposition}

\begin{proof}
Corollary~\ref{cor:binarygame} then applies once the local value of the GHZ game is identified as $3/4$.
The contradiction argument above shows that no deterministic local strategy can win on all four allowed inputs, while strategies saturating three of the four constraints exist.
Convexity then implies the same bound for all local mixtures.
\end{proof}

Proposition~\ref{prop:ghzgame} does not yet solve the full higher-party model-selection problem, but it removes a possible misunderstanding: the witness-to-KL-to-complexity pipeline is not intrinsically tied to bipartite CHSH.
In this respect CHSH mainly offers a convenient closed form and a simple empirical test bed.

\subsection{CHSH with arbitrary known design}

Consider a CHSH experiment with settings $(X,Y)\in\{0,1\}^2$ and outputs $(A,B)\in\{\pm1\}^2$.
Let $r_{xy}=\mathbb P(X=x,Y=y)$ be a known design distribution, not necessarily uniform.
For the standard sign pattern $g(0,0)=-1$ and $g=+1$ otherwise, define the win indicator
\begin{equation}
W:=\mathbf 1\{AB=g(X,Y)\}.
\end{equation}
Let $P_r$ denote the joint distribution on $(X,Y,A,B)$ induced by the design $r$ and the conditional table $P(a,b\mid x,y)$.
Define the design-averaged CHSH win rate
\begin{equation}
\omega_r(P):=\mathbb E_{P_r}[W].
\end{equation}

\begin{theorem}[CHSH witness lower-bounds the KL gap to locality]
\label{thm:chshdesign}
For every local model $Q$ with the same design distribution $r$,
\begin{equation}
D_{\mathrm{KL}}(P_r\Vert Q_r)
\ge
D_{\mathrm{KL}}\!\left(\mathrm{Bern}(\omega_r(P))\,\middle\|\,\mathrm{Bern}(\omega_r(Q))\right).
\end{equation}
Moreover,
\begin{equation}
\omega_r(Q)\le \omega_{\mathrm{loc}}(r):=1-\min_{x,y} r_{xy},
\label{eq:weightedlocalbound}
\end{equation}
so that
\begin{equation}
\inf_{Q\in\loc} D_{\mathrm{KL}}(P_r\Vert Q_r)
\ge
D_{\mathrm{KL}}\!\left(\mathrm{Bern}(\omega_r(P))\,\middle\|\,\mathrm{Bern}(\omega_{\mathrm{loc}}(r))\right).
\label{eq:chshgeneralbound}
\end{equation}
Under uniform inputs $r_{xy}=1/4$, one has $\omega_{\mathrm{loc}}(r)=3/4$ and
\begin{equation}
\omega_r(P)=\frac12+\frac18 S(P),
\end{equation}
where
$S(P)=|-E_{00}+E_{01}+E_{10}+E_{11}|$.
\end{theorem}

\begin{proof}
The first inequality is Corollary~\ref{cor:binarygame}.
For CHSH every deterministic local strategy loses on at least one setting and can choose which setting that is; therefore its weighted win probability is at most $1-\min_{x,y}r_{xy}$.
Taking convex hulls preserves this bound, giving Eq.~\eqref{eq:weightedlocalbound}.
For uniform inputs the standard relation between CHSH score and win probability yields the final statement.
\end{proof}

Theorem~\ref{thm:chshdesign} clarifies two points that are easy to blur in shorter notes.
First, the Bernoulli certificate is not intrinsically tied to uniform randomization; a known nonuniform design simply changes the local benchmark from $3/4$ to $1-\min r_{xy}$.
Second, nonuniform designs weaken the witness certificate by raising the local benchmark, which is why nearly uniform inputs are operationally valuable.
Figure~\ref{fig:designrobustness} visualizes this degradation.

Measured in bits per trial, the CHSH witness certificate is
\begin{equation}
\delta_{\mathrm{CHSH},r}(\omega)
:=
\frac{1}{\log 2}
D_{\mathrm{KL}}\!\left(\mathrm{Bern}(\omega)\,\middle\|\,\mathrm{Bern}(\omega_{\mathrm{loc}}(r))\right).
\label{eq:deltachsh}
\end{equation}
Combining this with Eq.~\eqref{eq:practicalcrossover} yields a conservative crossover rule.
Figure~\ref{fig:ncrit} shows the resulting $n_{\mathrm{crit}}$ curves for several complexity gaps $\Delta d=d_{\mathcal M}-d_{\loc}$.

\begin{figure}[t]
\centering
\includegraphics[width=\linewidth]{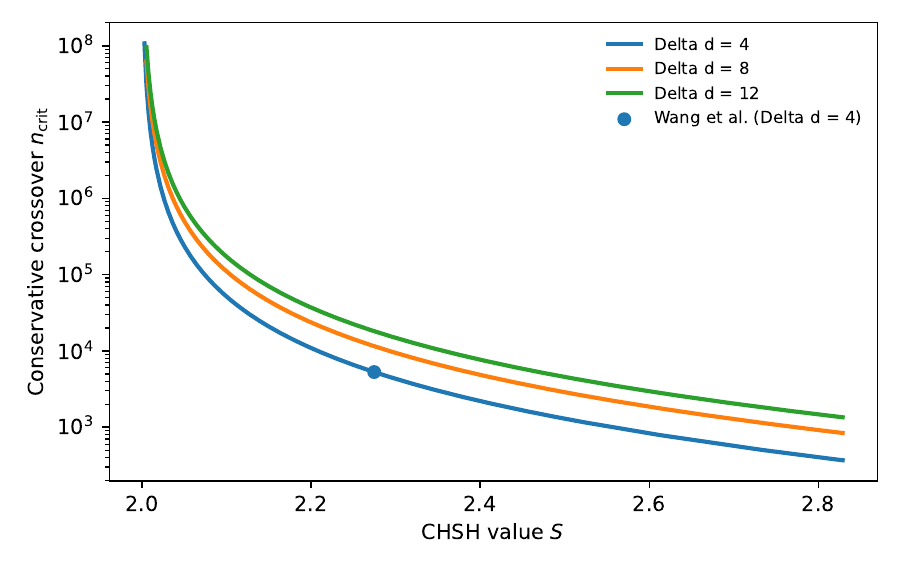}
\caption{Witness-certified crossover scales for several complexity gaps.
For each CHSH value $S$ we compute the Bernoulli KL lower bound
$\delta_{\mathrm{CHSH}}(S)=D_{\mathrm{KL}}(\mathrm{Bern}(\frac12+S/8)\Vert \mathrm{Bern}(3/4))/\log 2$
and then solve
$n\,\delta_{\mathrm{CHSH}}(S)\ge (\Delta d/2)\log_2 n$.
The three curves correspond to candidate competitors whose effective dimensions exceed the local polytope by $\Delta d\in\{4,8,12\}$.
The Wang point is marked for $\Delta d=4$, corresponding to the local-to-saturated comparison used later as the prototypical ``pay extra complexity'' benchmark.
}
\label{fig:ncrit}
\end{figure}

\begin{figure}[t]
\centering
\includegraphics[width=\linewidth]{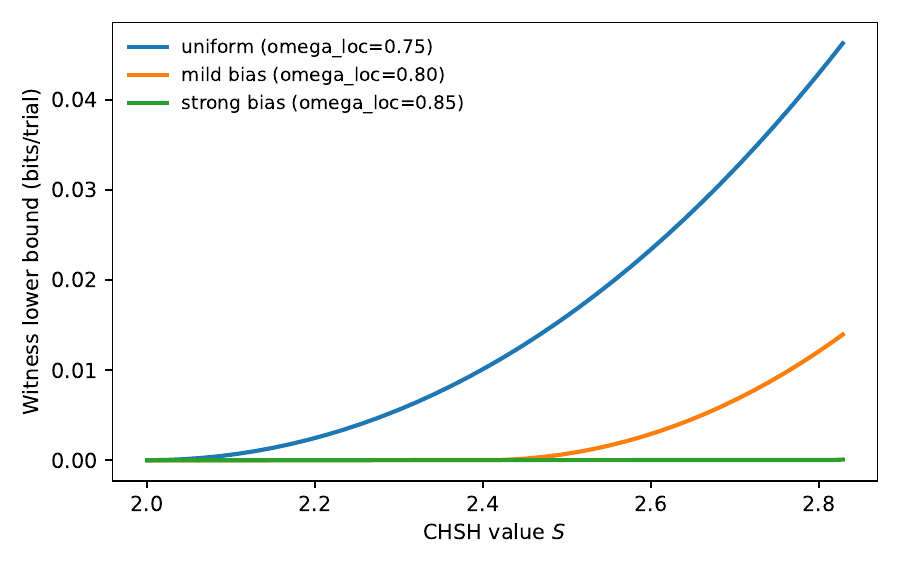}
\caption{Design robustness of the CHSH witness certificate.
The plotted curves show the witness lower bound in bits per trial as a function of $S$ for a uniform design and for two biased input designs.
As the smallest input weight decreases, the local benchmark $\omega_{\mathrm{loc}}(r)=1-\min r_{xy}$ rises and the same observed CHSH value certifies fewer bits.
Thus input imbalance weakens the witness-level evidence even before any full-table fit is performed.
}
\label{fig:designrobustness}
\end{figure}

\subsection{Finite-sample certificate and assumption boundary}

The witness certificate also admits a finite-sample confidence version.
The key observation is that Hoeffding's inequality only needs independence and boundedness, not identical distribution.
Thus a predetermined setting schedule is allowed as long as the design distribution is known.

\begin{proposition}[Finite-sample witness certificate for independent trials]
\label{prop:finitewitness}
Assume independent CHSH trials with known design $r$ and let $\hat\omega_n$ be the empirical mean of the win indicators.
For confidence level $1-\alpha$, define
\begin{equation}
\varepsilon_n(\alpha):=\sqrt{\frac{\ln(1/\alpha)}{2n}},
\qquad
\omega^-_\alpha:=\max\!\{\omega_{\mathrm{loc}}(r),\hat\omega_n-\varepsilon_n(\alpha)\}.
\end{equation}
Then with probability at least $1-\alpha$,
\begin{equation}
\inf_{Q\in\loc} D_{\mathrm{KL}}(P_r\Vert Q_r)
\ge
D_{\mathrm{KL}}\!\left(\mathrm{Bern}(\omega^-_\alpha)\,\middle\|\,\mathrm{Bern}(\omega_{\mathrm{loc}}(r))\right).
\label{eq:finitewitnesskl}
\end{equation}
\end{proposition}

\begin{proof}
Hoeffding's inequality gives
$\mathbb P(\omega_r(P)<\hat\omega_n-\varepsilon_n(\alpha))\le \alpha$.
Since the Bernoulli KL divergence to a fixed benchmark is nondecreasing for $p\ge \omega_{\mathrm{loc}}(r)$, the claim follows from Theorem~\ref{thm:chshdesign}.
\end{proof}

Proposition~\ref{prop:finitewitness} is theorem-level rigorous, but its application requires a genuine trial model.
This distinction matters in the case study below.
The appendix table of Wang \emph{et al.} records coincidence counts reconstructed from separate phase settings rather than a raw shot-by-shot event sequence.
Accordingly, the asymptotic witness rate is perfectly well defined for the induced conditional table, but the finite-sample statement can only be interpreted there as an \emph{effective-trial} approximation.
Theorem and application should therefore be read at different levels: the theorem certifies a finite-sample lower bound for independent trials with known design, whereas the public Wang table supports only an approximate back-of-the-envelope translation of that theorem.
That is not a flaw in the theorem; it is a limitation of the available public data representation.

\section{Empirical studies}

\subsection{Wang \emph{et al.} reanalysis and witness extraction}

We illustrate the framework on the publicly reported four-photon data of Wang \emph{et al.}, who observed a CHSH violation in a ``frustrated interference'' experiment with unentangled photons and listed the coincidence table in their appendix \cite{wang2025violation}.
Before giving numbers, it is worth making the theorem/application boundary explicit once more.
For this dataset, the asymptotic witness rate is rigorously defined for the reconstructed conditional table.
By contrast, the finite-sample confidence interpretation is only an effective-trial approximation because the public appendix reports blocked coincidence counts rather than an independent shot-by-shot trial stream.
The raw table is reproduced in Table~\ref{tab:wangcounts}.
Following the authors, the four outcomes for a setting pair $(\alpha,\beta)$ are reconstructed by a $\pi$ phase shift:
\begin{align}
n_{++\mid\alpha\beta}&:=N(\alpha,\beta), &
 n_{+-\mid\alpha\beta}&:=N(\alpha,\beta+\pi),\nonumber\\
n_{-+\mid\alpha\beta}&:=N(\alpha+\pi,\beta), &
 n_{--\mid\alpha\beta}&:=N(\alpha+\pi,\beta+\pi).
\label{eq:countsmap}
\end{align}

\begin{table}[t]
\caption{Four-fold coincidence counts $N(\alpha,\beta)$ (integration time $60\,\mathrm{s}$) reproduced from Ref.~\cite{wang2025violation}.
Columns correspond to $\beta_1=\pi/4$, $\beta_1+\pi=5\pi/4$, $\beta_2=3\pi/4$, and $\beta_2+\pi=7\pi/4$.
Rows correspond to $\alpha_1=0$, $\alpha_1+\pi=\pi$, $\alpha_2=\pi/2$, and $\alpha_2+\pi=3\pi/2$.}
\label{tab:wangcounts}
\begin{ruledtabular}
\begin{tabular}{ccccc}
 & $\beta_1$ & $\beta_1+\pi$ & $\beta_2$ & $\beta_2+\pi$\\
\hline
$\alpha_1$ & 77  & 295 & 271 & 63\\
$\alpha_1+\pi$ & 371 & 58  & 107 & 364\\
$\alpha_2$ & 331 & 139 & 333 & 54\\
$\alpha_2+\pi$ & 129 & 327 & 97  & 296\\
\end{tabular}
\end{ruledtabular}
\end{table}

From the reconstructed conditional table we obtain
\begin{equation}
\widehat S=2.274548,
\qquad
\sigma_{\widehat S}\approx 0.0567,
\end{equation}
where the uncertainty is estimated by Poisson Monte Carlo.
Equivalently,
\begin{equation}
\widehat\omega=\frac12+\frac{\widehat S}{8}\approx 0.784318.
\end{equation}
The witness-certified asymptotic information rate against locality is therefore
\begin{equation}
\delta_{\mathrm{CHSH}}(\widehat\omega)
\approx 0.004681\ \text{bits/trial}.
\label{eq:wangdeltahat}
\end{equation}
If the total coincidence count $n=3312$ is read as an effective trial count, this corresponds to about $15.5\bits$ of asymptotic witness evidence against the local polytope.
The 95\% finite-sample lower certificate from Proposition~\ref{prop:finitewitness} becomes
\begin{equation}
\omega^-_{0.05}\approx 0.763052,
\qquad
\delta^-_{\mathrm{CHSH}}\approx 6.63\times 10^{-4}\ \text{bits/trial},
\end{equation}
namely about $2.2\bits$ over the full table.
Because the public data are coincidence tables reconstructed from separate integrations, this finite-sample figure should be interpreted as an effective-trial illustration rather than a theorem-level confidence statement about the original experiment.

\subsection{Published witness-only checks beyond the Wang table}

To supplement the Wang coincidence-table analysis, we also compute witness-only certificates from several published CHSH experiments that report a central $S$ value but not a full conditional table suitable for reanalysis.
For these cases we use the standard uniform-design conversion $\omega=\frac12+S/8$ and evaluate the CHSH witness rate $\delta_{\mathrm{CHSH}}(S)$ from the reported central value alone.
This analysis is lighter than the Wang reanalysis: it is not a substitute for a full-table fit, but it does indicate how the witness-to-bits map behaves across distinct experimental platforms.

\begin{table*}[t]
\caption{Witness-only benchmarks from additional published CHSH experiments.
The witness rate is the uniform-design CHSH certificate $\delta_{\mathrm{CHSH}}(S)=D_{\mathrm{KL}}(\mathrm{Bern}(\frac12+S/8)\Vert \mathrm{Bern}(3/4))/\log 2$ evaluated at the reported central $S$ value.
These rows are external calibration points rather than full reanalyses, because the underlying full conditional tables are not reconstructed here. For Hensen \emph{et al.}, the listed value is the combined value reported in Ref.~\cite{hensen2016second}.}
\label{tab:publishedbenchmarks}
\begin{ruledtabular}
\begin{tabular}{lccc}
Experiment & Reported $S$ & Implied $\omega$ & $\delta_{\mathrm{CHSH}}$ (bits/trial)\\
\hline
Wang \emph{et al.} \cite{wang2025violation} & $2.2745$ & $0.7843$ & $0.004681$\\
Hensen \emph{et al.} (combined) \cite{hensen2016second} & $2.38\pm0.14$ & $0.7975$ & $0.009092$\\
Storz \emph{et al.} \cite{storz2023superconductingbell} & $2.0747\pm0.0033$ & $0.7593$ & $0.000338$\\
J\"ons \emph{et al.} \cite{jons2017bright} & $2.07\pm0.02$ & $0.7588$ & $0.000297$\\
\end{tabular}
\end{ruledtabular}
\end{table*}

Table~\ref{tab:publishedbenchmarks} makes two points that are easy to miss when reading only one case study.
First, the witness translation is portable across very different physical platforms: photons in nanowires, electron spins in diamond, superconducting circuits, and the four-photon interference setting of Wang \emph{et al.}
Second, the certified bits per trial vary nonlinearly with $S$ near the local boundary.
The Wang point sits in an intermediate regime: much stronger than the low-margin nanowire or superconducting-circuit examples in bits per trial, but below the combined diamond estimate.
Taken together, these examples indicate that the witness-based conversion behaves consistently across several experimental platforms.

\subsection{Model classes and why the cosine family is included}

We compare five model classes for the induced $2\times 2\times 2\times 2$ conditional table.

\paragraph*{Two-parameter cosine family $\qtwo$.}
The compact quantum candidate is
\begin{equation}
E_{V,\phi}(\alpha,\beta)=V\cos(\alpha+\beta+\phi),
\end{equation}
with unbiased marginals and probabilities
\begin{align}
p^{(Q_2)}_{++\mid \alpha\beta}=p^{(Q_2)}_{--\mid \alpha\beta}=\frac{1+E_{V,\phi}(\alpha,\beta)}{4},\nonumber\\
p^{(Q_2)}_{+-\mid \alpha\beta}=p^{(Q_2)}_{-+\mid \alpha\beta}=\frac{1-E_{V,\phi}(\alpha,\beta)}{4}.
\label{eq:q2model}
\end{align}
This family is not chosen because it is guaranteed to win a complexity contest.
It is included because it is the smallest smooth family that captures two experimentally meaningful deformations of the ideal CHSH geometry: a global visibility loss $V$ and a common phase offset $\phi$.
In symmetric photonic interference settings these are the first physically interpretable parameters one would write down.
The parameter $V$ absorbs imperfect interference contrast, admixture, and other visibility-reducing effects into a single physically meaningful scale.
The parameter $\phi$ captures a common analyzer misalignment or birefringence-induced phase offset across the four settings, which is a natural coherent deformation to consider before introducing setting-specific correlator distortions.
In that sense $\qtwo$ is meant to be a ``compact but meaningful'' first quantum comparator: not the only low-dimensional nonlocal family one could write down, but the natural first one for an experiment whose geometry is designed around a nearly sinusoidal correlator pattern.

\paragraph*{Four-parameter unbiased-correlator control $\qfour$.}
To test whether the data really support the extra cosine structure, we also fit a less structured zero-marginal family with one independent correlator $E_{xy}$ per setting pair.
This control is still low-dimensional and nonlocal, but it removes the shared trigonometric relation among settings.
If $\qtwo$ beat only locality while losing badly to this control, then the apparent success of the cosine family would be too easy to overstate.

\paragraph*{Reference and broader controls.}
The local polytope $\loc$ is the foundational benchmark, the no-signalling polytope $\ns$ is a broader nonlocal control, and the saturated conditional model $\sat$ is the predictive ceiling that fits each setting independently.
The saturated class is useful in the crossover discussion because it is genuinely more expressive than locality; Eq.~\eqref{eq:practicalcrossover} is most informative for such comparisons.

\subsection{Full-table comparison, low-dimensional control, and bootstrap stability}

Table~\ref{tab:modelcomparison} summarizes the fitted code lengths.
The two-parameter cosine family has fitted parameters
\begin{equation}
\widehat V=0.807016,
\qquad
\widehat\phi=3.014843\ \mathrm{rad}.
\end{equation}
Among the five fitted classes, $\qtwo$ has the smallest BIC.
The local, no-signalling, and saturated controls are respectively worse by about $21.10\bits$, $5.74\bits$, and $9.93\bits$.
For comparison, the broader four-parameter unbiased-correlator control is only $0.35\bits$ worse than $\qtwo$.
The data therefore support a compact nonlocal description over locality, while only weakly distinguishing whether the additional cosine structure of $\qtwo$ is warranted.

\begin{table*}[t]
\caption{Complexity-charged comparison for the Wang \emph{et al.} table.
Here $d$ denotes the nominal effective dimension, $-\log_2\hat{\mathcal L}$ is the fitted negative log-likelihood in bits, and $\Delta\mathrm{BIC}$ is measured relative to the two-parameter cosine family $\qtwo$.}
\label{tab:modelcomparison}
\begin{ruledtabular}
\begin{tabular}{lcccc}
Model & $d$ & $-\log_2\hat{\mathcal L}$ (bits) & $L_{\mathrm{BIC}}$ (bits) & $\Delta\mathrm{BIC}$ vs $\qtwo$ (bits)\\
\hline
$\qtwo$ (cosine, unbiased marginals) & 2  & 5802.884 & 5814.578 & 0.000\\
$\qfour$ (one correlator per setting, unbiased marginals) & 4 & 5791.536 & 5814.923 & 0.345\\
$\loc$ (local polytope) & 8 & 5788.899 & 5835.673 & 21.095\\
$\ns$ (no-signalling polytope) & 8 & 5773.547 & 5820.321 & 5.743\\
$\sat$ (saturated conditional model) & 12 & 5754.348 & 5824.509 & 9.931\\
\end{tabular}
\end{ruledtabular}
\end{table*}

To probe stability, we bootstrap the table within each setting pair while keeping the original setting totals fixed.
Across $1000$ resamples,
\begin{equation}
\Pr[\mathrm{BIC}(\qtwo)<\mathrm{BIC}(\loc)]\approx 0.953,
\end{equation}
whereas
\begin{equation}
\Pr[\mathrm{BIC}(\qtwo)<\mathrm{BIC}(\qfour)]\approx 0.488.
\end{equation}
So the low-dimensional nonlocal-versus-local conclusion is stable, but the exact two-parameter-vs-four-parameter distinction is not.
This is the relevant distinction for interpreting the role of the cosine ansatz.

It is also useful to check how much of the ranking is specific to BIC's growing complexity surcharge.
Using the same fitted likelihoods with Akaike's constant-penalty criterion \cite{akaike1974new}, we obtain
\begin{align}
L_{\mathrm{AIC}}(\qtwo)&\approx 5805.77, &
L_{\mathrm{AIC}}(\qfour)&\approx 5797.31,\nonumber\\
L_{\mathrm{AIC}}(\loc)&\approx 5800.44, &
L_{\mathrm{AIC}}(\ns)&\approx 5785.09,\nonumber\\
L_{\mathrm{AIC}}(\sat)&\approx 5771.66
\end{align}
(bits throughout).
Thus AIC prefers the broader nonlocal controls, especially the saturated model.
This does not change the main empirical conclusion, but it helps delimit its scope.
Across both BIC and AIC, locality is disfavored.
What changes with the complexity proxy is whether one prizes parsimonious low-dimensional nonlocal structure or out-of-sample predictive flexibility.
This sensitivity is one reason to state the witness theorem independently of any single proxy and to read the full-table conclusions as complexity-charged summaries rather than as exact evidence statements.

The full-table gain of the best local fit over the saturated reference corresponds to
\begin{equation}
\frac{-\log_2\hat{\mathcal L}_{\loc}+\log_2\hat{\mathcal L}_{\sat}}{n}
\approx 0.01043\ \text{bits/trial},
\end{equation}
which exceeds the witness lower bound in Eq.~\eqref{eq:wangdeltahat} by a factor of about $2.2$.
This is the expected direction: the witness discards information, so its role is to certify a conservative floor rather than to reproduce the full-table gain exactly.

\subsection{Simulation calibration: how conservative is the witness bound?}

The Wang reanalysis is only one data table, so we add a controlled simulation study.
We generate synthetic CHSH-like tables in three regimes---weak violation, Wang-like violation, and strong violation---by starting from a CHSH-symmetric nonlocal table and then drawing each setting distribution from a Dirichlet neighborhood around that base table.
This construction breaks the exact cosine structure while preserving the overall nonlocal scale.
For each synthetic table we compute the witness lower bound and the actual full-table gain over the best local fit.

Figure~\ref{fig:calibration} shows the result.
All points lie on or above the diagonal, as they must if the witness is a lower bound.
More interesting is how the gap behaves with violation size.
Near the local boundary the bound is markedly conservative: the median full-table gain is about $6.2$ times the median witness lower bound in the weak-violation regime.
Around the Wang-like regime the median ratio drops to about $1.4$, and for stronger violations it is close to $1.1$.
Thus the witness is most conservative near the local boundary and becomes appreciably tighter once the violation is farther from that boundary.

\begin{figure}[t]
\centering
\includegraphics[width=\linewidth]{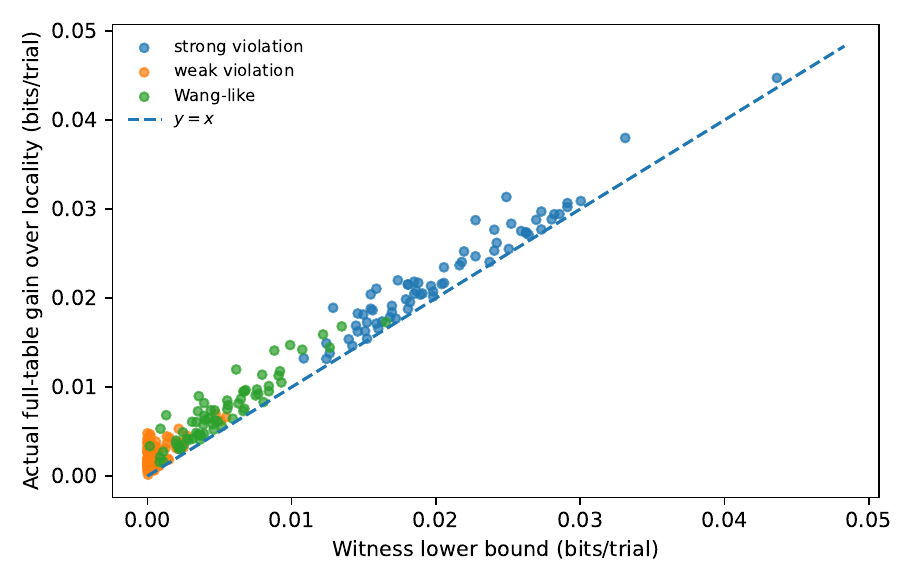}
\caption{Calibration of the witness lower bound against the actual full-table gain over locality.
Each point is a synthetic CHSH-like table generated from a Dirichlet-perturbed nonlocal distribution in one of three regimes.
The horizontal axis is the witness-certified KL lower bound in bits per trial; the vertical axis is the fitted full-table gain of the saturated model over the best local fit, also in bits per trial.
The dashed line is $y=x$.
The witness is most conservative near the local boundary and becomes substantially tighter for stronger violations.
}
\label{fig:calibration}
\end{figure}

\subsection{Phase diagram for paying extra complexity}

A natural setting in which to examine Eq.~\eqref{eq:practicalcrossover} is the comparison between locality and a genuinely more expressive control.
We therefore simulate CHSH-like nonlocal tables and ask when the saturated model beats locality by BIC.
This comparison carries an extra complexity gap of $\Delta d=4$, so the witness crossover line is obtained from
\begin{equation}
n\,\delta_{\mathrm{CHSH}}(S)\gtrsim 2\log_2 n.
\end{equation}

Figure~\ref{fig:phase} shows the resulting phase diagram.
The color indicates the frequency with which the saturated model beats locality by BIC over repeated synthetic tables at a given $(n,S)$.
The dashed line is the witness-predicted crossover from the lower bound alone.
As expected, the witness line is conservative: the full-table transition occurs somewhat earlier, because the full table exploits structure beyond the single win-rate statistic.
Nevertheless, the witness line tracks the correct trend and gives the right order of magnitude.
In particular, around the Wang-like scale $S\approx 2.28$, the empirical $50\%$ BIC transition occurs near $n\approx 3200$, which is close to the Wang table size of $3312$ effective trials.
This calibration helps place the Wang case study on a broader quantitative scale.

\begin{figure}[t]
\centering
\includegraphics[width=\linewidth]{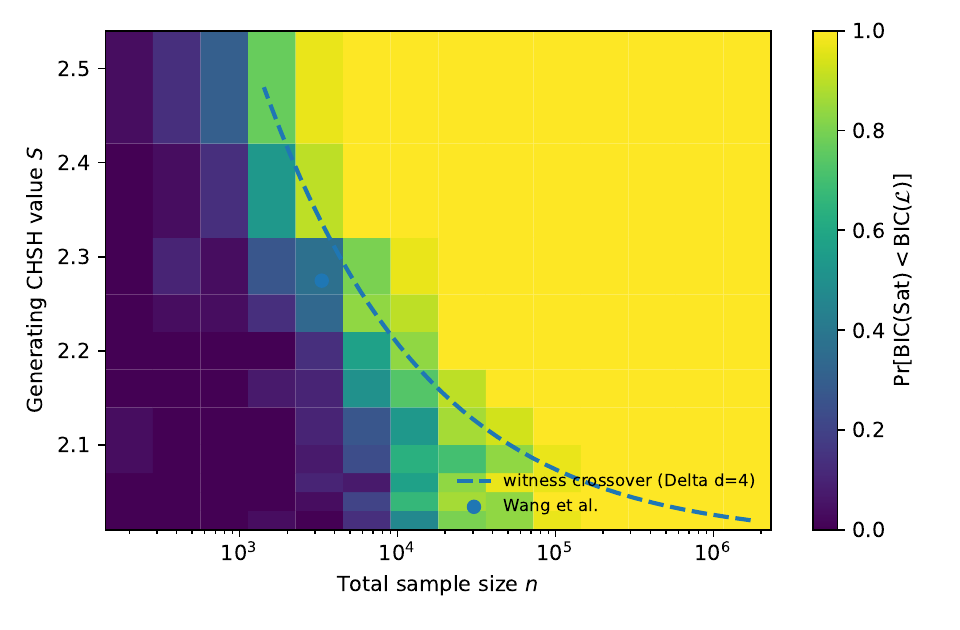}
\caption{Simulation phase diagram for the local-versus-saturated crossover.
Synthetic CHSH-like tables are generated from perturbed nonlocal distributions, and the color indicates the frequency with which the saturated model beats locality by BIC.
The dashed line is the witness-only crossover prediction for the extra complexity gap $\Delta d=4$.
The witness line is conservative but captures the correct transition scale.
The Wang point is marked for reference.
}
\label{fig:phase}
\end{figure}

\section{Interpretational scope of the BIC comparison}
\label{sec:bicscope}

The local and no-signalling classes in Table~\ref{tab:modelcomparison} are constrained polyhedral models.
Their likelihood geometry depends on which face or boundary region is active at the optimum, so the textbook regular-model derivation of BIC does not automatically apply \cite{watanabe2009algebraic,drton2017singular}.
This is also the main reason one should keep theorem-level witness claims and heuristic full-table comparisons conceptually separate.

In the present paper the separation is explicit.
Theorem~\ref{thm:coarsegraining}, Corollary~\ref{cor:binarygame}, Theorem~\ref{thm:chshdesign}, and Proposition~\ref{prop:finitewitness} are rigorous and do not rely on BIC at all.
Their output is a certified lower bound on KL separation, either asymptotically or at fixed confidence under the stated sampling assumptions.
By contrast, the BIC numbers in Table~\ref{tab:modelcomparison} should be read as MDL-style summaries of fit versus nominal dimension.
They are useful because they are transparent, reproducible, and easy to compare in bits, not because they provide exact marginal-likelihood asymptotics for polytope classes.

This distinction changes how the results should be interpreted.
The rigorous statement is that the CHSH witness alone certifies a positive information gap against locality.
The empirical statement is that, on this data table, complexity-charged full-table fits favor low-dimensional nonlocal descriptions over locality and over broader controls.
The bootstrap and simulation studies help clarify the scope of the empirical statement.
They show that the local-versus-nonlocal conclusion is stable, that the witness lower bound is conservative in a quantifiable way, and that the finer distinction between two closely related compact nonlocal families is not stable.

The AIC calculation in Sec.~V.D makes the same point from another angle.
A weaker, $n$-independent penalty moves the preference toward broader nonlocal classes, but it does not restore a preference for locality.
So the robust conclusion is ``nonlocality beats locality''; what depends on the complexity bill is whether the preferred nonlocal account is a compact structured family or a more weakly regularized predictive control.

In a longer program one would want sharper MDL tools tailored to constrained or singular Bell-model classes, such as singular-BIC, normalized maximum-likelihood surrogates, or explicit predictive-validation scores.
Nothing in the witness theory prevents that substitution.
The framework is modular in the following sense: once a better complexity proxy for polytope classes is adopted, the same witness-certified KL rate can be compared to that proxy instead of to classical BIC.

\section{Discussion and outlook}

The main conclusion is that Bell witnesses can be converted into conservative information rates in the same units used by complexity penalties, which allows witness-based evidence to be related directly to complexity-charged model comparison.
CHSH provides the most transparent closed-form instance of this relation.

Several features of the analysis are worth separating.
First, the witness construction is not limited to CHSH: any binary Bell-game witness inherits the Bernoulli certificate of Corollary~\ref{cor:binarygame}, and Proposition~\ref{prop:ghzgame} gives a concrete multipartite example.
Second, the Wang data support compact nonlocal descriptions over locality, but they do \emph{not} strongly favor the specific two-parameter cosine ansatz over the broader four-parameter unbiased-correlator control.
Third, the additional witness-based benchmarks show that the translation from reported CHSH values to certified bits per trial remains informative across several experimental platforms.
Fourth, the AIC comparison and the simulation studies identify where the witness and full-table criteria differ: the witness is most conservative near the local boundary, and the preferred nonlocal model depends on how strongly additional flexibility is penalized.

Several extensions remain of interest.
On the Bell side, it would be useful to obtain closed-form or efficiently computable witness-image characterizations beyond CHSH, especially for multipartite nonlocal games and multivalued scores.
On the MDL side, one would like complexity proxies adapted to singular and polyhedral classes.
On the empirical side, the framework should be tested on additional public Bell datasets with raw trial streams rather than reconstructed coincidence tables, so that the finite-sample theorem can be used directly rather than approximately.
These extensions would sharpen the quantitative analysis, but they do not alter the main observation of the paper: a witness-certified KL lower bound can already be compared meaningfully with a complexity penalty when deciding whether extra model complexity is warranted.

\bibliography{references}

\end{document}